\documentclass{ifacconf}

\usepackage{ifluatex,ifxetex}
\if\ifluatex T\else\ifxetex T\else F\fi\fi T 
	\usepackage{fontspec}
    \ExplSyntaxOn 
    \msg_redirect_name:nnn { unicode-math } { mathtools-overbracket } { none }
    \msg_redirect_name:nnn { unicode-math } { mathtools-colon } { none }
    \ExplSyntaxOff%
	\usepackage{unicode-math}
\else
	\usepackage[T1]{fontenc}
	\usepackage[utf8]{inputenc}
	\usepackage{unicodehelper}
\fi

\usepackage{hotfixIFAC}
\usepackage{microtype}

\usepackage{xparse}

\usepackage{graphicx}
\usepackage{natbib}
\usepackage{url}

\usepackage{amsmath}

\let\digamma\relax

\usepackage{amssymb}
\def\qedsymbol{\ensuremath{\square}}
\usepackage{mathtools}

\usepackage{fp}
\usepackage{tikz}
\usetikzlibrary{arrows,arrows.meta,calc,positioning,fpu,decorations.pathmorphing,decorations.pathreplacing}
\usepackage{pgfplots}\pgfplotsset{compat=1.14}

\usepackage{subfiles}

\usepackage{flushend}

\usepackage{hyperref}%

\usepackage{cleveref}

\newcommand{\ie}{i\/.\/e\/.,\/\ }%
\newcommand{\eg}{e\/.\/g\/.,\/\ }%
\newcommand{\cf}{cf\/.\/\ }%
\newcommand{\iid}{%
\renewcommand{\iid}{i\/.\/i\/.\/d\/.\ }%
independent identically distributed (i\/.\/i\/.\/d\/.)\ %
}%
\NewDocumentCommand{\abvAss}{s}{\IfBooleanTF{#1}{Ass.\,}{Ass.\,}}
\NewDocumentCommand{\abvThm}{s}{\IfBooleanTF{#1}{Thms.\,}{Thm.\,}}
\NewDocumentCommand{\abvLem}{s}{\IfBooleanTF{#1}{Lemmata~}{Lemma~}}
\NewDocumentCommand{\abvProp}{s}{\IfBooleanTF{#1}{Props.\,}{Prop.\,}}
\begin{document}
\begin{frontmatter}

\title{
A Constraint-Tightening Approach to\\
\leavevmode\kern0pt\clap{Nonlinear Stochastic Model Predictive Control}\kern0pt\\
for Systems under General Disturbances\kern5pt$^{\ref{sponsor}}$
}

\thanks[sponsor]{%
The authors thank the German Research Foundation (DFG) for financial support under Grants
AL 316/12-2, and MU 3929/1-2; %
and the International Max Planck Research School for Intelligent Systems (IMPRS-IS)
for their support of the first author.%
\newline%
The authors are with the Institute for Systems Theory and Automatic Control, University of Stuttgart, Germany (e-mail: \{\href{mailto:hennnig.schlueter@ist.uni-stuttgart.de}{schlueter}, \href{mailto:frank.allgoewer@ist.uni-stuttgart.de}{allgoewer}\}@\href{https://www.ist.uni-stuttgart.de}{ist.uni-stuttgart.de}).%
}

\author{Henning Schlüter\qquad} 
\author{Frank Allgöwer}

\begin{abstract}                
This paper presents a nonlinear model predictive control strategy for stochastic systems with general (state and input dependent) disturbances subject to chance constraints.
Our approach uses an online computed stochastic tube to ensure stability, constraint satisfaction and recursive feasibility in the presence of stochastic uncertainties.
The shape of the tube and the constraint backoff is based on an offline computed incremental Lyapunov function. 
\end{abstract}

\begin{keyword}
Predictive control, Constrained control, Stochastic control, Nonlinear control
\end{keyword}

\end{frontmatter}

\section{Introduction}
Model Predictive Control (MPC) \citep{Mayne2014} is a widely-used optimization-based control method, which is able to handle general nonlinear constrained systems.
For nominal MPC schemes, which are assuming that an actual deterministic model of the system is available, rigorous theoretical guarantees (such as recursive feasibility, constraint satisfaction and stability) are well established in the literature \citep{Rawlings2017}.
Robust and stochastic MPC (RMPC and SMPC, respectively) have been developed to ensure these properties despite uncertainties in the model and/or external disturbances \citep{Kouvaritakis2016}.
While RMPC generally assumes that uncertainties lie in bounded sets, SMPC can additionally incorporate stochastic descriptions.
This enables SMPC to enforce \emph{chance constraints}, which are constraints that allow for a given probability of violation.

In many domains, stochastic models for complex phenomena, \eg loads or failures in electrical power grids, are well-established, yet these phenomena often arise in already nonlinear control problems. 
In order to tackle such problems, we propose an SMPC framework for nonlinear systems with rigorous theoretical guarantees.
Existing SMPC approaches for nonlinear systems \citep{Schildbach2014} suffer from a tremendous amount of online computation.
Our method on the other hand is able to consider nonlinear systems under general disturbances at the price of only a limited increase in online computational demand over nominal MPC scheme.

\subsection*{Related work}
\citet{Mesbah2016} summarizes the current state of the art of SMPC and notes that there is a lack of efficient algorithms for nonlinear systems that are able to consider general probabilistic uncertainty descriptions.
In this work, we aim to provide such an algorithm in tradition of tube-based approaches to SMPC, which are among the most efficient methods.

Tube-based solutions to propagate uncertainty were first proposed for RMPC \citep{Chisci2001} for linear systems.
This has later been extended to nonlinear systems using class $\mathcal{K}$ functions or Lipschitz constants \citep{Pin2009}.
Such approaches were shown to be conservative, especially for longer prediction horizons, and are often difficult to implement for nonlinear systems.
This method was extended by \citet{Santos2019} to the stochastic case.
The authors were able to precompute a constraint backoff for the chance constraints enforced at the first step, since they only considered the additive disturbances.
From there, the uncertainty could be propagated as in any RMPC approach, as they eliminated the uncertainty early in the approach.
In this article, we consider general uncertainty, that may also depend on the current state and input, hence we have to consider uncertainty in the online optimization.

\cite{Villanueva2017} proposed to compute the tube fully online employing min-max-differential-inequalities.
This leads to a complexity increase over nominal MPC in the number of states squared.

A middle ground in online complexity is achieved in \cite{Koehler2019}, where, by using sublevel sets of an incremental Lyapunov function (ILF) as tube, the authors reduce the conservationism significantly, while only requiring a single additional state and constraint over nominal MPC.
Our method on the other hand just introduces a single constraint for each chance constraint probability considered, enabling stochastic disturbances.

Inspired by these result, we propose an extension of the computationally efficient framework by \cite{Koehler2019} to SMPC, which is additionally able to consider stochastic disturbances and chance constraints.

\subsection*{Notation}
The quadratic norm with respect to a positive definite matrix $Q ≻ 0$ is denoted by $‖x‖^2_Q=x^⊤ Q x$, the minimal and maximal eigenvalue of Q are denoted by $λ_\mathrm{min}$ and $λ_\mathrm{max}$, respectively. 
The positive real numbers are $ℝ_{≥0}=\left\lbrace r∈ℝ\middle|r≥0\right\rbrace$. 
$\mathcal{K}_∞$ denotes the class of functions $α: ℝ_{≥0} → ℝ_{≥0}$, which are continuous, strictly increasing, unbounded and satisfy $α(0)=0$.  
The probability of an event $s\in\mathcal{S}$ is $ℙ[s\in\mathcal{S}]$ and the expected value of a random variable $s$ is $𝔼[s]$.
When conditioned on a time-index $t$, they are denoted by $ℙ_t$ and $𝔼_t$, respectively.
If the time argument is not stated explicitly, $x^+$ denotes $x(t+1)$, while $x$ is used for $x(t)$.
A nominal prediction for time step $t+k$ based on the state at time $t$ is denoted with index $k|t$, \eg $x_{k|t}$.
\section{Preliminaries}
\subsection{Problem setup}
We consider a nonlinear stochastic discrete-time system
\begin{equation}
	x(t+1) = f_w(x(t),u(t),d(t)) \label{eq:sys:full}
\end{equation}
with time $t\in ℕ$, state $x\in ℝ^n$, control input $u\in ℝ^m$, and bounded \iid random variables $d(t)\in 𝔻$ as disturbance.
The nominal prediction model is chosen by certainty-equivalence as
\begin{equation}
	x^+ = f(x,u)≔𝔼_t\!\left[f_w(x,u,d)\right] . \label{eq:sys:nom}
\end{equation}
Thus, the system can be decomposed into
\begin{equation}
	x^+ = f_w(x,u,d) ≕ f(x,u)+d_w(x,u) \label{eq:sys:w}
\end{equation}
with the model mismatch $d_w$ as random variable.

Firstly, we enforce hard state and input constraints
\begin{equation}
	(x(t),u(t)) ∈ \mathcal{Z}_\mathrm{R} \label{eq:hc}
\end{equation}
with some compact nonlinear constraint set 
\begin{equation*}
	\mathcal{Z}_\mathrm{R}=\left\lbrace(x,u)\in ℝ^{n+m}\middle|g_j(x,u)\le 0, j=1,\dotsc,q_\mathrm{R}\right\rbrace ⊆ ℝ^{n+m}\,.
\end{equation*}
Secondly, we impose nonlinear individual chance constraints (ICC) on an output at the next time step, \ie
\begin{equation}
	ℙ_t[h_j(x(t+1),u(t+1))≤0]≥p_j, \quad j=1,\dotsc,q_\mathrm{P}\,. \label{eq:iccs}
\end{equation}
with a probability level $p_j∈(0,1)$. The set of all probability levels used by at least one of the ICCs is denoted 
\begin{equation}
\mathcal{P}≔\left\lbrace p_j \middle| j=1,\dotsc,q_\mathrm{P}\right\rbrace.
\end{equation}

Instead of requiring the exact cumulative distribution function, we make use of a lower bound thereupon, which may be easier to obtain in practice.
\begin{assumption}\label[assumption]{ass:disturbance:bound:certain}\label[assumption]{ass:disturbance:icdf}
		The random variable $d_w$ \cref{eq:sys:w} has a known probability distribution $p_w(x,u)$ with compact finite support $\mathcal{W}(x,u)$ for all $(x,u)\in \mathcal{Z}_\mathrm{R}$.
		Hence, for any $ε∈[0,1]$, there exits a scalar function $\hat{w}^ε:\mathcal{Z}_\mathrm{R} → ℝ_{\ge0}$ that satisfies 
		\begin{equation} \label{eq:disturbance:bound:w_hat}
			ℙ[\left\lVert d_w(x,u)\right\rVert≤\hat{w}^ε(x,u)]≥ε
		\end{equation}
		with $\hat{w}^ε(x,u)$ finite for all $x$ and $u$.
		Furthermore, $\hat{w}^ε$ satisfies the following monotonicity property:
		\begin{equation}
			\bgroup\newcommand{\nemu}{\mkern-1mu}\newcommand{\nemumu}{\mkern-2mu}
			∀\,\nemu(x,u) \nemu∈\nemu \mathcal{Z}_\mathrm{R}, 0≤\nemu ε_1\nemu≤\nemu ε_2\nemu≤\nemu1 : \hat{w}^{ε_1}\nemumu(x,u)\nemu≤\nemu\hat{w}^{ε_2}\nemu(x,u)\,.
			\egroup
		\end{equation}
\end{assumption}

This uncertainty description encompasses additive, multiplicative and more general nonlinear disturbances or unmodeled nonlinearities.

We assume that $f(0,0)=0$ and that the constraints satisfy
\begin{equation*}
	0 ∈ \operatorname{int} (\mathcal{Z}_\mathrm{R} ∩ \left\lbrace(x,u)\in ℝ^{n+m}\middle|h_j(x,u) ≤ 0, j=1,\dotsc,q_\mathrm{P}\right\rbrace),
\end{equation*} 
since we consider the problem of stabilizing the origin.
Further, the control objective is to minimize the open-loop cost $J_N$ of the predicted state and input sequence, with
\begin{equation}
	J_N(x_{⋅|t},u_{⋅|t}) = \sum_{k=0}^{N-1} ℓ(x_{k|t},u_{k|t})+V_f(x_{N|t})\,,
\end{equation}
where the stage cost $ℓ$ and terminal cost $V_f$ (defined in \cref{sec:scheme:terminal}) are positive definite.

\subsection{Local incremental stabilizability}
In order to provide the theoretical guarantees for the constraint backoff and robust stability, it is assumed, similarly to \citet{Koehler2019}, that the system is locally incrementally stabilizable. For completeness, we restate the required assumptions and adapt them, where necessary, to accommodate the chance constraints.
\begin{assumption} \citep[Ass.\/\,2]{Koehler2019}\label[assumption]{ass:local_incremental_stabilizability}
	There exist a control law $κ : ℝ^n × \mathcal{Z}_\mathrm{R} \rightarrow ℝ^m$ , an incremental Lyapunov function (ILF) $V_δ : ℝ^n × \mathcal{Z}_\mathrm{R} \rightarrow ℝ_{≥0}$, which is continuous in the first argument and satisfies $V_δ(z, z, v) = 0$ for all $(z, v) \in \mathcal{Z}_\mathrm{R}$, and parameters $c_{δ,l}, c_{δ,u} , δ_{\mathrm{loc}} , κ_{\mathrm{max}} > 0, ρ ∈ (0, 1)$, such that the following	properties hold for all $(x, z, v) ∈ ℝ^n × \mathcal{Z}_\mathrm{R}$ with $V_δ(x, z, v) ≤ δ_{\mathrm{loc}}$, and all $(x^+, z^+, v^+) ∈ ℝ^n × \mathcal{Z}$:
	\begin{align}
			c_{δ,l}\left\lVert x - z\right\rVert^2 ≤ V_δ(x, z, v) &≤c_{δ,u}\left\lVert x - z\right\rVert^2, \label{eq:local_incremental_stabilizability:lyapbound}\\
			\left\lVert κ(x, z, v) - v\right\rVert^2 &≤ κ_\mathrm{max} V_δ(x, z, v)\,, \label{eq:local_incremental_stabilizability:input}\\
			V_δ(x^+, z^+, v^+) &≤ ρ^2 V_δ(x, z, v)\,, \label{eq:local_incremental_stabilizability:contractivity}
	\end{align}
	with $x^+ = f (x, κ(x, z, v))$, and $z^+ = f(z, v)$.
\end{assumption}
The ILF will be used to construct the stochastic tube later on, yet we only require its existence and knowledge of the scalar parameters, but not the functions $V_δ$, $κ$ themselves. 
In particular, we exploit the fact that the ILF provides an upper bound on the achievable contraction rate between two trajectories, \eg between the predicted trajectory of the MPC scheme and the closed-loop trajectory.

The following assumptions enable us to compute scalar bounds that relate the nonlinear constraints \cref{eq:hc} and \cref{eq:iccs} to the level sets of the ILF $V_δ$.
\begin{assumption} \citep[Ass.\/\,3]{Koehler2019} \label[assumption]{ass:cost:kinf_bound}
	The stage cost $ℓ: \mathcal{Z}_\mathrm{R} → R≥0$ satisfies
	\begin{align}
		ℓ(r) &≥ α_ℓ(\left\lVert r\right\rVert)\,,\\
		ℓ(\tilde{r}) - ℓ(r) &≤ α_c(\left\lVert r\right\rVert)\,,\quad ∀r∈\mathcal{Z}, \tilde{r}∈ℝ^{n+m},
	\end{align}
	with $α_ℓ, α_c ∈\mathcal{K}_∞$. Furthermore, for any $ρ ∈ (0, 1)$, we have $α_{c,ρ}(c) ≔ \sum_{k=0}^{∞} α_c(ρ^kc)∈\mathcal{K}_∞$.
\end{assumption}
\begin{assumption}\label[assumption]{ass:constraints:lipschitz}\label[assumption]{ass:icc:lipschitz}\label[assumption]{ass:hc:lipschitz}
	There exist local Lipschitz constants $L^\mathrm{R}_i$, $L^\mathrm{P}_j$, such that
	\begin{align}
		g_i(\tilde{r}) - g_i(r) &≤ L^\mathrm{R}_i\left\lVert r - \tilde{r} \right\rVert,\quad i = 1,\dotsc,q_\mathrm{R}\,,\\
		h_i(\tilde{r}) - h_i(r) &≤ L^\mathrm{P}_i\left\lVert r - \tilde{r} \right\rVert,\quad i = 1,\dotsc,q_\mathrm{P}\,,
	\end{align}
	holds for all $r ∈ \mathcal{Z}_\mathrm{R}$ and all $\tilde{r}∈ ℝ^{n+m}$ with $\left\lVert r - \tilde{r}\right\rVert^2 ≤ \frac{δ_{\mathrm{loc}}}{c_{δ,l}}$.
\end{assumption}
\begin{proposition}\label[proposition]{prop:incremental_bound}\label[proposition]{prop:incremental_bound:hc}\label[proposition]{prop:incremental_bound:cost}\label[proposition]{prop:incremental_bound:icc}
	Suppose that \cref{ass:local_incremental_stabilizability,ass:cost:kinf_bound,ass:constraints:lipschitz} hold, then there exist constants $c^\mathrm{R}_i≥0, i=1,\dotsc,q_\mathrm{R}$,  $c^\mathrm{P}_j≥0, j=1,\dotsc,q_\mathrm{P}$, and a function $α_u ∈\mathcal{K}_∞$ such that the following inequalities hold for all $(x, z, v) ∈ℝ^n × \mathcal{Z}$ with $V_δ(x, z, v) ≤ c^2$ and any $c ∈ [0, δ_{\mathrm{loc}}]$:
	\bgroup\allowdisplaybreaks
	\begin{alignat}{5}
		ℓ(x, κ(x, z, v)) &- {}&{} ℓ\,(z, v) &≤{}\span\span{} α_u(c)\,,\label{eq:incremental_bound:cost}\\
		g_j(x, κ(x, z, v)) &- {}&{} g_j(z, v) &≤ c^\mathrm{R}_j&&\mkern-1mu\cdot  c\,,\label{eq:incremental_bound:hc}\\
		h_j(x, κ(x, z, v)) &- {}&{} h_j(z, v) &≤ c^\mathrm{P}_j&&\mkern-1mu\cdot  c\,.\label{eq:incremental_bound:icc}
	\end{alignat}\egroup
\end{proposition}
\begin{proof}
	For the proof of the first part, \ie \cref{eq:incremental_bound:cost} and \cref{eq:incremental_bound:hc}, see \citet[Prop.\/\,1]{Koehler2019}. \Cref{eq:incremental_bound:icc} is derived analogously to \cref{eq:incremental_bound:hc}.
	\hfill\qedsymbol
\end{proof}

This proposition will allow us to relate the constraints to the tube, we construct in the next sections.

\subsection{Efficient uncertainty description}
Additionally, for the tube construction, we need to consider how uncertainty propagation affects the ILF.
A computationally efficient way, proposed by \citet{Koehler2019}, is to describe the uncertainty in terms of the ILF.
As we not only consider bounded set disturbances, but also stochastic uncertainties, we need an revised construction.
\begin{assumption} \label[assumption]{ass:disturbance:icdf:ilyap}
	Consider the disturbance bound $\hat{w}$, the incrementally stabilizing feedback $κ$ and the ILF $V_δ$ from \cref{ass:disturbance:icdf,ass:local_incremental_stabilizability}.
	For any $ε∈[0,1]$, there exists a function $w^ε_δ : \mathcal{Z} × ℝ_{≥0} → ℝ_{≥0}$, such that for any point $(x, z, v) ∈ℝ^n × \mathcal{Z}$ with $V_δ(x, z, v) ≤ c^2$, and any $c ∈ [0, δ_\mathrm{loc} ]$, we have 
	\begin{align}
		\hat{w}^ε(x, κ(x, z, v)) &≤ w^ε_δ(z, v, c)\,.\label{eq:disturbance:icdf:ilyap}
	\shortintertext{%
	Furthermore, $w_δ$ satisfies the following monotonicity properties: Firstly, for any point $(x, z, v) ∈ ℝ^n × \mathcal{Z}$ such that $V_δ(x, z, v) ≤ (c_1 - c_2)^2$ with constants $0 ≤ c_2 ≤ c_1 ≤ δ_\mathrm{loc}$, we have
	}
		w^ε_δ(x, κ(x, z, v), c_2) &≤ w^ε_δ(z, v, c_1)\,.\label{eq:disturbance:icdf:ilyap:monoton:c}
	\shortintertext{%
	Secondly, for any constant $0≤ε_1≤ε_2≤1$, we have
	}
		w^{ε_1}_δ(x, κ(x, z, v), c) &≤ w^{ε_2}_δ(z, v, c)\,.\label{eq:disturbance:icdf:ilyap:monoton:delta}
	\end{align}
\end{assumption}

This assumption establishes $ω^ε_δ$ as an $ε$-likely upper bound on the uncertainty that can occur at a state $x$ of an incrementally stabilized trajectory in a neighborhood of a point $(z,v) ∈ \mathcal{Z}_\mathrm{R}$, where the neighborhood is given by $V_δ(x,z,v)≤c^2$. 
Based thereupon, we can bound the increase of the ILF due to the disturbance in the next time step with probability $ε$.

\begin{proposition}\label[proposition]{prop:disturbance:ilyap}
	Let \cref{ass:disturbance:icdf,ass:local_incremental_stabilizability,ass:disturbance:icdf:ilyap} hold.
	Then, there exists a function $\tilde{w}^ε_δ : \mathcal{Z} × ℝ_{≥0} → ℝ_{≥0}$, such that for any point
	$(x, z, v) ∈ ℝ^n × \mathcal{Z}$ with $V_δ(x, z, v) ≤ c^2$ , any $c ∈ [0, δ_\mathrm{loc}]$,
	any $(z^+, v^+) ∈ \mathcal{Z}$ with $z^+ = f(z, v)$, and disturbance $d_w$ as random variable, we have 
	\begin{equation}
		ℙ\!\left[V_δ(z^++d_w(x,κ(x,z,v)), z^+, v^+) ≤ (\tilde{w}^ε_δ(z, v, c))^2\right]≥ε\,.\label{eq:disturbance:ilyap}
	\end{equation}
	Furthermore, $\tilde{w}^ε_δ$ satisfies the same monotonicity properties as $w^ε_δ$, \ie\cref{eq:disturbance:icdf:ilyap:monoton:c} and \cref{eq:disturbance:icdf:ilyap:monoton:delta} hold for $\tilde{w}^ε_δ$.
\makeatletter
\protected@write\@auxout{}{%
\string\newlabel{eq:disturbance:ilyap:monoton:c}{{\ref{eq:disturbance:icdf:ilyap:monoton:c}}{\pageref{eq:disturbance:icdf:ilyap:monoton:c}}{Efficient uncertainty description}{equation.2.\ref{eq:disturbance:icdf:ilyap:monoton:c}}{}}
\string\newlabel{eq:disturbance:ilyap:monoton:c@cref}{{[equation][\ref{eq:disturbance:icdf:ilyap:monoton:c}][]\ref{eq:disturbance:icdf:ilyap:monoton:c}}{\pageref{eq:disturbance:icdf:ilyap:monoton:c}}}
\string\newlabel{eq:disturbance:ilyap:monoton:delta}{{\ref{eq:disturbance:icdf:ilyap:monoton:delta}}{\pageref{eq:disturbance:icdf:ilyap:monoton:delta}}{Efficient uncertainty description}{equation.2.\ref{eq:disturbance:icdf:ilyap:monoton:delta}}{}}
\string\newlabel{eq:disturbance:ilyap:monoton:delta@cref}{{[equation][\ref{eq:disturbance:icdf:ilyap:monoton:delta}][]\ref{eq:disturbance:icdf:ilyap:monoton:delta}}{\pageref{eq:disturbance:icdf:ilyap:monoton:delta}}}
}
\makeatother
\end{proposition}
\begin{proof}
	The proof follows trivially from the assumptions, by setting $\tilde{w}^ε_δ(z, v, c)=\sqrt{c_{δ,u}}w^ε_δ(z,v,c)$.\hfill\qedsymbol
\end{proof}
The function $\tilde{w}_δ^ε$ can be constructed similarly as in \cite{Koehler2019}, an example there of is given in \cref{sec:numsim}.

In the absence of chance constraint, we could now construct the tube as in \citet{Koehler2019}.
In fact, by setting $ε=1$, this reduces to the same considerations and results.
This is how we will implement the hard constraints \cref{eq:hc}. 
For the chance constraints \cref{eq:iccs}, however, additional consideration are required, in order to ensure closed-loop constraint satisfaction, which are discussed later in \cref{sec:scheme:icc}.

\section{Stochastic Model Predictive Control Framework}
This section presents the proposed stochastic MPC framework for nonlinear uncertain systems.
The overall scheme is introduced in \cref{sec:scheme:framework}.
In \cref{sec:scheme:icc} the constraint backoff for the chance constraints are discussed.
The theoretical analysis in \cref{sec:scheme:theory} uses the terminal ingredients described in \cref{sec:scheme:terminal}.

\subsection{Proposed nonlinear MPC scheme}\label{sec:scheme:framework}
The basic idea of our scheme is similar to \citet{Koehler2019}.
Therein, the authors proposed to indirectly characterize the tube as sublevel sets of the ILF $V_δ$ (\cref{ass:local_incremental_stabilizability}) by an online predicted tube size $s^1$.
Then, this tube size is used to tighten the state and input constraints ensuring robust constraint satisfaction.
In this work, this is extended to allow for ICCs and stochastic uncertainties. 
Here, satisfaction is ensured by designing a constraint backoff.
This depends on the state $x$ and the input $u$ to accommodate the $(x,u)$-dependence of the disturbance, as well as on the robust tube size $s^1$, in order to account for uncertainty accumulated over the previous steps of the prediction.
This backoff introduces additional probabilistic tube sizes $s^p$.

\begin{samepage}
This lead to the deterministic optimization problem
\begin{subequations}\label{eq:snmpc}
\begin{align}
	V_N(x(t)) &= \min_{u_{\cdot |t},w^\mathrm{R}_{\cdot |t},w^p_{\cdot |t}} J_N(x_{\cdot |t}, u_{\cdot |t})\\	
	\mathrm{s.t. }\; &x_{0|t} = x(t),\quad s^p_{0|t} = 0,\label{eq:snmpc:init}\\
	&x_{k+1|t} = f (x_{k|t},u_{k|t}),\label{eq:snmpc:state}\\
	&s^p_{k+1|t} = ρs^1_{k|t}+w^p_{k|t},\label{eq:snmpc:var:sp}\\
	&w^p_{k|t} ≥ \tilde{w}_δ^p (x_{k|t}, u_{k|t}, s^p_{k|t}),\label{eq:snmpc:var:wp}\\
	&h_j (x_{k+1|t}, u_{k+1|t})+c^\mathrm{P}_j s^p_{k+1|t} ≤ 0,\label{eq:snmpc:constraint:chance}\\
	&g_i (x_{k|t}, u_{k|t})+c^\mathrm{R}_i s^1_{k|t} ≤ 0,\label{eq:snmpc:constraint:robust}\\
	&s^\mathrm{1}_{k|t} ≤ \bar{s},\quad w^p_{k|t} ≤ w^\mathrm{1}_{k|t} ≤ \bar{w},\label{eq:snmpc:constraint:tube}\\
	&(x_{N|t}, s^1_{N|t}) ∈ \mathcal{X}_{\!\!f},\label{eq:snmpc:constraint:terminal}\\
	&i = 1,\dotsc, q_\mathrm{R},\quad j = 1,\dotsc, q_\mathrm{P},\notag\\
    &k = 0,\dotsc, N - 1,\quad  p∈\mathcal{P}∪\lbrace 1\rbrace,\notag
\end{align}
\end{subequations}
which is to be solved at each time instant. %
\end{samepage}%
The solution of \cref{eq:snmpc} are optimal trajectories for the state $x^*_{\cdot |t}$, the input $u^*_{\cdot |t}$, the tube sizes $s^{p,*}_{\cdot |t}$, the disturbance bounds $w^{p,*}_{\cdot |t}$, and the value function $V_N$.
The terminal ingredients $V_f$, $\mathcal{X}_{\mkern-5mu f}$, $\bar{s}$, and $\bar{w}$ are introduced in \cref{sec:scheme:terminal}. 

The first portion of the resulting optimal input sequence is applied to the system, resulting in the closed-loop system is given by
\begin{equation}
    x(t+1) = f(x(t),u(t),d(t))\,,\qquad u(t)≔u^*_{0|t}\,.
\end{equation}

\subsection{Chance Constraints}\label{sec:scheme:icc}
For the sake of simplicity, we will consider in this section without loss of generality only single ICCs
\begin{equation}
	ℙ_t[h(x(t+1),u(t+1))≤0]≥p\,. \label{eq:icc:single}
\end{equation}

In the literature, the chance constraints are commonly handled by so-called constraint backoffs. 
This idea originates in linear MPC with additive stochastic disturbances \citep{vanHessem2002}.
There, one can simply \emph{backoff} the constraint, by enforcing at least a precomputed constant distance from the constraint boundary.  
In this work, however, we consider nonlinear systems with general disturbances, where the required backoff not only becomes state-dependent, but also intractable to compute.

Using the sublevel sets of the ILF, we can construct a tube around the prediction, which contains the disturbed closed-loop trajectory with at least probability $p$.
The size of this tube will be used as our backoff.

\begin{figure}
    \centering
    \begin{tikzpicture}[
		x=2cm,y=-2cm,
    	point/.style={draw,fill=white,rectangle,minimum height=1pt,rounded corners=0.5pt,inner sep=0cm},
		area/.style={draw,fill=white,rectangle,inner sep=0cm,rounded corners=2.5pt},
		phead/.style={draw,rectangle,minimum height=0pt,inner sep=0cm},
        rhead/.style={fill},
        inwardtube/.style={decorate, decoration={snake, amplitude=0.75pt, segment length=6pt},opacity=0.5},
        robusttube/.style={orange,line cap=round,very thick},
        probtube/.style={blue,line cap=round,very thick},
        ctube/.style={draw=green!50!black,text=green!50!black,line cap=round},
        prediction/.style={black,thick}
	]
	\pgfdeclarelayer{background}
	\pgfdeclarelayer{foreground}
	\pgfdeclarelayer{backgroundA}
	\pgfdeclarelayer{mainA}
	\pgfdeclarelayer{foregroundA}
	\pgfsetlayers{background,backgroundA,mainA,foregroundA,main,foreground}
	\begin{pgfonlayer}{background}
		\foreach \x in {0.5,...,2.5} \draw[dotted,line cap=round] (\x,-0.3) -- (\x,2.2);
		\foreach \x in {0,...,3} \node[] at (\x,2.1) {$k=\x$};
	\end{pgfonlayer}
	%
	\begin{scope}[minimum width=15pt]
		\begin{pgfonlayer}{foregroundA}
			\coordinate[point,ctube] (x0) at (0,0.25);
		\end{pgfonlayer}
		\begin{pgfonlayer}{foregroundA}
			\coordinate[point,ctube] (x1) at (1,0.5);
			\coordinate (x1north) at ($(x1)-(0,0.7)$);
			\coordinate (x1south) at ($(x1)+(0,0.7)$);
			\coordinate (x1northP) at ($(x1)-(0,0.5)$);
			\coordinate (x1southP) at ($(x1)+(0,0.5)$);
		\end{pgfonlayer}
		\begin{pgfonlayer}{backgroundA}
			\draw[robusttube] (x0.north) -- (x1north);
			\draw[robusttube] (x0.south) -- (x1south);
			\draw[prediction] (x0.east) -- (x1.west);
		\end{pgfonlayer}
		\begin{pgfonlayer}{mainA}
			\path[rhead] (x1north) circle (1pt);
            \path[rhead] (x1south) circle (1pt);
			\draw[very thick,probtube] (x1northP) node[phead] {} -- (x1southP) node[phead] {};
		\end{pgfonlayer}
		\begin{pgfonlayer}{foregroundA}
			\coordinate[area,minimum height=1cm,ctube] (x2) at (2,0.75);
			\coordinate (x2north) at ($(x2.north)-(0,0.7)$);
			\coordinate (x2south) at ($(x2.south)+(0,0.7)$);
			\coordinate (x2northP) at ($(x2.north)-(0,0.5)$);
			\coordinate (x2southP) at ($(x2.south)+(0,0.5)$);
		\end{pgfonlayer}
		\begin{pgfonlayer}{backgroundA}
			\draw[robusttube] (x1north) -- (x2north);
			\draw[robusttube] (x1south) -- (x2south);
			\draw[inwardtube] (x1north) -- (x2.north);
			\draw[inwardtube] (x1south) -- (x2.south);
			\draw[prediction] (x1.east) -- (x2.west);
		\end{pgfonlayer}
		\begin{pgfonlayer}{mainA}
        	\path[rhead] (x2north) circle (1pt);
            \path[rhead] (x2south) circle (1pt);
			\draw[very thick,probtube] (x2northP) node[phead] {} -- (x2southP) node[phead] {};
		\end{pgfonlayer}
		\begin{pgfonlayer}{foregroundA}
			\coordinate[area,minimum height=1.5cm,ctube] (x3) at (3,0.875);
			\coordinate (x3north) at ($(x3.north)-(0,0.7)$);
			\coordinate (x3south) at ($(x3.south)+(0,0.7)$);
			\coordinate (x3northP) at ($(x3.north)-(0,0.5)$);
			\coordinate (x3southP) at ($(x3.south)+(0,0.5)$);
            \coordinate (x4north) at ($(x3north)+(1,0)$);
            \coordinate (x4south) at ($(x3south)+(1,0.2)$);
		\end{pgfonlayer}
		\begin{pgfonlayer}{backgroundA}
			\draw[robusttube] (x2north) -- (x3north);
			\draw[robusttube] (x2south) -- (x3south);
			\draw[inwardtube] (x2north) -- (x3.north);
			\draw[inwardtube] (x2south) -- (x3.south);
			\draw[prediction] (x2.east) -- (x3.west);
            \begin{scope}
                \clip (-0.5,-0.3) rectangle (3.5,2.2);
                \draw[robusttube] (x3north) -- (x4north);
                \draw[robusttube] (x3south) -- (x4south);
            \end{scope}
		\end{pgfonlayer}
		\begin{pgfonlayer}{mainA}
			\path[rhead] (x3north) circle (1pt);
            \path[rhead] (x3south) circle (1pt);
			\draw[very thick,probtube] (x3northP) node[phead] {} -- (x3southP) node[phead] {};
		\end{pgfonlayer}
	\end{scope}

	\begin{pgfonlayer}{foreground}
		\draw[prediction] (x0.west) -- (x0.east);
		\draw[prediction] (x1.west) -- (x1.east);
		\draw[prediction] (x2.west) -- (x2.east);
		\draw[prediction] (x3.west) -- (x3.east);
		\path[opacity=0.5] (x1north) -- node [anchor=south,pos=0.65] {$\cdot ρ$} (x2.north);
		\path[opacity=0.5] (x2north) -- node [anchor=south,pos=0.65] {$\cdot ρ$} (x3.north);
		\begin{scope}[draw opacity=0.6]
			\draw[probtube,semithick,decorate,decoration={brace,amplitude=5pt,aspect=0.25}] ($(x1southP)+(-0.5pt,0)$) -- node[anchor=east,pos=0.25,xshift=-2pt] {$\mathcal{S}^p_{1|t}$} ($(x1northP)+(-0.5pt,0)$);
			\draw[robusttube,semithick,decorate,decoration={brace,amplitude=5pt,aspect=0.675}] ($(x1north)+(0.5pt,0)$) -- node[anchor=west,pos=0.675,xshift=2pt] {$\:\mathcal{S}^1_{1|t}$} ($(x1south)+(0.5pt,0)$);
			\draw[probtube,semithick,decorate,decoration={brace,amplitude=5pt,aspect=0.13}] ($(x2southP)+(-0.5pt,0)$) -- node[anchor=east,pos=0.13,xshift=-2pt] {$\mathcal{S}^p_{2|t}$} ($(x2northP)+(-0.5pt,0)$);
			\draw[robusttube,semithick,decorate,decoration={brace,amplitude=5pt,aspect=0.79}] ($(x2north)+(0.5pt,0)$) -- node[anchor=west,pos=0.79,xshift=2pt] {$\:\mathcal{S}^1_{2|t}$} ($(x2south)+(0.5pt,0)$);
			\draw[probtube,semithick,decorate,decoration={brace,amplitude=5pt,aspect=0.11}] ($(x3southP)+(-0.5pt,0)$) -- node[anchor=east,pos=0.11,xshift=-2pt] {$\mathcal{S}^p_{3|t}$} ($(x3northP)+(-0.5pt,0)$);
			\draw[robusttube,semithick,decorate,decoration={brace,amplitude=5pt,aspect=0.82}] ($(x3north)+(0.5pt,0)$) -- node[anchor=west,pos=0.82,xshift=2pt] {$\:\mathcal{S}^1_{3|t}$} ($(x3south)+(0.5pt,0)$);
			
			\draw[draw opacity=1,white,ultra thick] (2.5,0.78) -- (2.5,0.475);
			\draw[ctube,semithick,decorate,decoration={brace,amplitude=5pt,aspect=0.3}] (x2.north east) -- node[anchor=west,pos=0.3] {$\:ρ\mathcal{S}^1_{1|t}$} (x2.south east);
			\draw[ctube,semithick,decorate,decoration={brace,amplitude=5pt,aspect=0.3}] (x3.north east) -- node[anchor=west,pos=0.3] {$\:ρ\mathcal{S}^1_{2|t}$} (x3.south east);
		\end{scope}
	\end{pgfonlayer}
\end{tikzpicture}
    \caption{
        Illustration of the idea behind the proposed incremental backoff (\cref{thm:backoff:incremental}).
        The robust and stochastic tubes are shown in orange and blue, respectively.
    }
    \label{fig:backoff:incremental}
\end{figure}
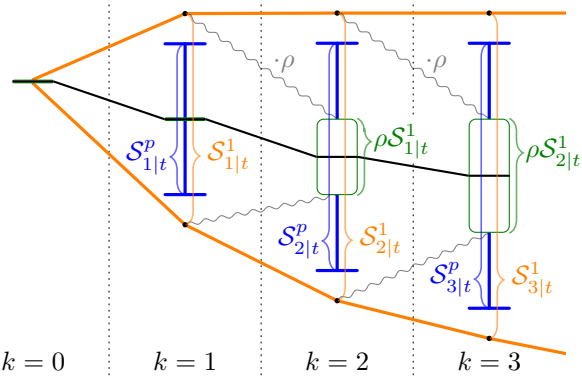
The idea of the tube construction is illustrated in \cref{fig:backoff:incremental}.
Starting off, we begin with the prediction (black).
Using \cref{prop:disturbance:ilyap} for $ε=1$, a robust tube (orange) can be constructed around this prediction, inside which the true state will certainly lie.
This tube is constituted by the sublevel set
\begin{equation}
	\mathcal{S}^1_{k|t}≔\left\lbrace x∈ℝ^n\middle|V_δ(x,x_{k|t},u_{k|t})≤s^1_{k|t}\right\rbrace\,.
\end{equation}

If one assumes that the previous time step was without disturbance, \ie $d_w=0$, then a contraction $ρ$ of the robust set (\cref{ass:local_incremental_stabilizability}) is reached by the incremental stabilization $κ$.
Thus, we obtain an inner tube (green) lacking the influence of the last disturbance with the sets $\mathcal{S}^0_{k+1|t}≔ρ\mathcal{S}^1_{k|t}$.

Using \cref{prop:disturbance:ilyap} for $ε=p∈(0,1)$, the disturbance is added to the tube. 
Thereby, we obtain an $ε$-likely tube, indicate by the blue error bars.
This tube confines the state with a probability greater than $ε$ at each time step in the sets
\begin{align}
	&\mathcal{S}^p_{k+1|t}≔\mathcal{S}^1_{k|t}+\left\lbrace x∈ℝ^n\middle|V_δ(x,x_{k|t},u_{k|t})≤w^p_{k|t}\right\rbrace\\
	&\quad=\left\lbrace x∈ℝ^n\middle|V_δ(x,x_{k|t},u_{k|t})≤ρs^1_{k|t}+w^p_{k|t}≕s^p_{k+1|t}\right\rbrace\,.\notag
\end{align}

By this construction, we can employ the size $s^p$ of the sublevel sets of $V_δ$, \ie the size of our tube, as our backoff to ensure ICC satisfaction formally by the following theorem.
\begin{theorem}\label[theorem]{thm:backoff:incremental}
	Suppose that \cref{ass:disturbance:icdf,ass:local_incremental_stabilizability,ass:icc:lipschitz,ass:disturbance:icdf:ilyap} hold, then
    for all $k≥0$ with $(x_{k-1|t+1}, x_{k|t}, u_{k|t})∈ℝ^n×\mathcal{Z}$ such that $V_δ(x_{k-1|t+1}, x_{k|t}, u_{k|t})≤(s^1_{k|t})^2$, the deterministic constraints \cref{eq:snmpc:constraint:chance} imply the nonlinear ICCs \cref{eq:iccs}.
\end{theorem}\vspace*{-0.3\baselineskip}
This theorem will be proven jointly with \cref{thm:snmpc} in \cref{sec:scheme:theory}.

\subsection{Terminal ingredients}\label{sec:scheme:terminal}
By using the minimal bound on the uncertainty $\bar{w}_\mathrm{min}$ and the maximal tube size $\bar{s}$
\begin{align}
	\bar{w}_\mathrm{min} &= \inf_{(x,u)∈\mathcal{Z}_\mathrm{R}} \tilde{ω}_δ(x,u,0)\,, & \bar{s}&=\sqrt{δ_\mathrm{loc}}\,,\label{eq:snmpc:smax}
\end{align}
we capture the desired properties of the terminal ingredients in the following assumption.
\begin{assumption}\label[assumption]{ass:snmpc:terminalcontroller}
	There exist a terminal controller $k_f : ℝ^n → ℝ^m$, a terminal cost function $V_f : ℝ^n → ℝ_{≥0}$, a terminal set $\mathcal{X}_{\mkern-5mu f} ⊂ ℝ^{n+1}$ , and a constant $\bar{w} ∈ ℝ_{≥0}$ such that the following holds for all $(x, s) ∈ \mathcal{X}_{\mkern-5mu f}$ ,  all $d_w ∈ ℝ^n$,  all $w ∈ [\bar{w}_\mathrm{min}, \bar{w}]$, and all $s^+ ∈
	[0, ρs - ρ^N w+\tilde{w}_δ(x, k_f(x), s)]$, such that $V_δ(x^++d_w, x^+, k_f(x^+)) ≤ ρ^{2N} w^2$ with $x^+ = f(x, k_f(x))$:
	\begin{subequations}\label{eq:snmpc:terminalcontroller}
	\begin{align}
		V_f (x) - ℓ(x, k_f(x)) &≥ V_f(x^+) \,,\label{eq:snmpc:terminalcontroller:ilyap}\\
		(x^++d_w , s^+) &∈ \mathcal{X}_{\mkern-5mu f}\,,\label{eq:snmpc:terminalcontroller:xf}\\
		\tilde{w}_δ^1(x, k_f(x), s) &≤ \bar{w}\,,\label{eq:snmpc:terminalcontroller:w}\\
		g_i(x, k_f(x))+c^\mathrm{R}_is &≤ 0,\label{eq:snmpc:terminalcontroller:hc}\\
		ρs-ρ^Nw+\tilde{w}_δ^{p_j}(x, k_f(x), s)&≕β^{p_j}\\
		h_j(x^+, k_f(x^+))+c^\mathrm{P}_jβ^{p_j} &≤ 0,\label{eq:snmpc:terminalcontroller:icc}\\
		s &≤ \bar{s}\,,\label{eq:snmpc:terminalcontroller:s}
	\end{align}
	\end{subequations}
	with $i = 1,\dotsc,q_\mathrm{R}$, and $j = 1,\dotsc,q_\mathrm{P}$.
	Furthermore, the terminal cost $V_f$ is continuous on the compact set $\mathcal{X}_{f, x} ≔ \left\lbrace x\,\middle|\, ∃\mkern1mu s ∈ [0, s], (x, s) ∈ \mathcal{X}_{\mkern-5mu f} \right\rbrace$, i.e., there
	exists a function $α_f ∈ \mathcal{K}_∞$ such that
	\begin{equation}
		V_f(z) ≤ V_f(x)+α_f(‖x - z‖), ∀x, z ∈ \mathcal{X}_{f, x}\,. \label{eq:snmpc:terminalcontroller:ilyap:continuous}
	\end{equation}
\end{assumption}
These technical conditions are similar to the standard conditions in nominal MPC for the augmented state $(x,s)$ and input $(u,w,w^p)$.
Details on constructive satisfaction can be found in \citet{Koehler2019}.

The only extension to the robust case, is the inclusion of the ICC in the construction of the terminal set, \ie \cref{eq:snmpc:terminalcontroller:icc}. 
This ensures that also the ICCs are satisfied by the terminal controller, using the same backoff technique as just described in \cref{sec:scheme:icc}.
For the construction of the terminal set, these constraints are treated similar to the hard constraint \cref{eq:snmpc:terminalcontroller:hc}.

\subsection{Theoretical analysis}\label{sec:scheme:theory}
In the following theorem, we provide guarantees on the closed-loop properties of the proposed MPC scheme.
\begin{theorem}\label[theorem]{thm:snmpc}
	Let \cref{ass:local_incremental_stabilizability,ass:cost:kinf_bound,ass:hc:lipschitz,ass:disturbance:icdf,ass:icc:lipschitz,ass:snmpc:terminalcontroller,ass:disturbance:icdf:ilyap} hold, and suppose that \cref{eq:snmpc} is feasible at t = 0. Then \cref{eq:snmpc} is recursively feasible, the constraints \cref{eq:hc}, \cref{eq:iccs} are satisfied and the origin is practically asymptotically stable for the resulting closed loop system.
\end{theorem}
\begin{proof}[Proof of \cref{thm:snmpc}]\begingroup%
	\renewcommand{\theparagraph}{\roman{paragraph}}%
	\RenewDocumentCommand{\paragraph}{s m}{\par\stepcounter{paragraph}\textsc{\theparagraph.}\,\textit{#2}\/\kern0.5pt:\hspace*{0.1em}\IfBooleanF{#1}{\hspace*{0.4em}}}%
	The proof is based on an extension of the main idea behind \citet[\abvThm1]{Koehler2019}, as such we will refer to their results, whenever it is possible.
	This will enable us to focus on handling the chance constraints, as the impact of the hard constraints is equivalent.
	
	The core idea is to use the control law $κ$ from \cref{ass:local_incremental_stabilizability} to construct a candidate solution, ensuring recursive feasibility and bounding the cost increase.
	
	\paragraph{Candidate Solution}\bgroup\allowdisplaybreaks
	For convenience, define
	\begin{subequations}
	\begin{align}
		u^∗_{N|t} &=k_f(x^∗_{N|t})\,, \quad u^∗_{N+1|t} = k_f(x^∗_{N+1|t}),\\
		x^∗_{N+1|t} &= f(x^∗_{N|t}, u^∗_{N|t})\,,\\
		w^{p,∗}_{N|t} &=\tilde{w}_δ^p(x^∗_{N|t}, u^∗_{N|t}, s^∗_{N|t})\,.
	\end{align} 
	\end{subequations}
	Consider the adapted candidate solution, \ie
	\begin{subequations}\label{eq:snmpc:recursive_feasbility:proof:candiate}
	\begin{align}
		x_{0|t+1} &= x(t+1) = f_w(x_{0|t}, u_{0|t}, d(t))\,,\\
		u_{k|t+1} &= κ(x_{k|t+1}, x^∗_{k+1|t}, u^∗_{k+1|t})\,,\\
		x_{k+1|t+1} &= f(x_{k|t+1}, u_{k|t+1})\,,\\
		s^p_{k+1|t+1} &= ρ s^1_{k|t+1}+w^p_{k|t+1}\,,\quad s^p_{0|t+1} = 0\,,\\
		w^p_{k|t+1} &= \tilde{w}_δ^p(x_{k|t+1}, u_{k|t+1}, s_{k|t+1})\,,  \label{eq:snmpc:recursive_feasbility:proof:candiate:wp}
	\end{align}
	\end{subequations}
	with $k=0,\dotsc,N-1$ and $p∈\mathcal{P}$.
	As in \citet[eq.\,17]{Koehler2019}, we obtain using \cref{prop:disturbance:ilyap} \cref{eq:disturbance:ilyap} with $ε=1$ and repeatedly applying \cref{ass:local_incremental_stabilizability} \cref{eq:local_incremental_stabilizability:contractivity} that for $k = 0,\dotsc,N$
	\begin{equation}
		V_δ(x_{k|t+1}, x^*_{k+1|t}, u^*_{k+1|t}) ≤ ρ^{2k} [w^*_{0|t}]^2 ≤ δ_{\mathrm{loc}}\,.	\label{eq:snmpc:recursive_feasbility:proof:Vdelta:bound}
	\end{equation}
	Thus, the candidate and previous optimal solution stay in the region $V_δ(z,x,v)≤ δ_\mathrm{loc}$, for which we have a local incremental Lyapunov function $V_δ$ by \cref{ass:local_incremental_stabilizability}.
	\egroup
	
	\paragraph{Tube Dynamics} From \citet[Proof of Thm.\,1, Part II, eq.\,18-19]{Koehler2019}, we have the inequalities
	\begin{align}
		s^1_{k|t+1} &≤ s^{1,*}_{k+1|t} - ρ^k w^{1,*}_{0|t}\,,\label{eq:snmpc:recursive_feasbility:proof:s1:decrease}\\
		w^1_{k|t+1} &≤ w^{1,*}_{k+1|t}\,,\label{eq:snmpc:recursive_feasbility:proof:w1:decrease}
	\end{align}
	for $k=0,\dotsc,N-1$. 
	Analogously to the derivation of \cref{eq:snmpc:recursive_feasbility:proof:w1:decrease}, we can show
	\begin{equation}
		w^p_{k|t+1}≤w^{p,*}_{k+1|t}\,.
		\label{eq:snmpc:recursive_feasbility:proof:wp:decrease}
	\end{equation}
	
	This enables us to consider the general case of $s^p_{k|t+1}$, yielding that for all $p ∈ \mathcal{P}∪\lbrace1\rbrace$ the inequality 
	\begin{equation}
		s^p_{k|t+1}≤s^{p,*}_{k+1|t}-ρ^{k}w^{1,*}_{0|t}
		\label{eq:snmpc:recursive_feasbility:proof:sp:decrease}\\
	\end{equation}
	holds for all $k=0,\dotsc,N-1$ and $j=1,\dotsc,q_\mathrm{P}$, since
	\begin{align*}
		\begingroup\arraycolsep=1pt\begin{array}{rcl}
			s^p_{0|t+1} 
			&{}\stackrel{\cref{eq:snmpc:init}}{=}{}&
			0\stackrel{\cref{eq:snmpc:var:sp}}{=}
			s^{1,*}_{1|t}-ρ^{0}w^{1,*}_{0|t}\\
			s^p_{k+1|t+1} 
			&{}\stackrel{\cref{eq:snmpc:var:sp}}{=}{}&
			ρs^1_{k|t+1}+w^p_{k|t+1}\\
			&{}\stackrel{\cref{eq:snmpc:recursive_feasbility:proof:s1:decrease}}{≤}{}&
			ρs^{1,*}_{k+1|t}-ρ^{k+1}w^{1,*}_{0|t}+w^p_{k|t+1}\\
			&{}\stackrel{\cref{eq:snmpc:recursive_feasbility:proof:wp:decrease}}{≤}{}&
			ρs^{1,*}_{k+1|t}-ρ^{k+1}w^{1,*}_{0|t}+w^{p,*}_{k+1|t}\\
			&{}\stackrel{\cref{eq:snmpc:var:sp}}{=}{}&
			s^{p,*}_{k+2|t}-c^\mathrm{P}_jρ^{k+1}w^{1,*}_{0|t}\,.
		\end{array}\endgroup
	\end{align*}

	\paragraph{Satisfaction of Hard Constraints, Terminal Constraints, and Tube Bounds} 
	The constraints \cref{eq:snmpc:constraint:robust,eq:snmpc:constraint:terminal,eq:snmpc:constraint:tube}
	are satisfied by the candidate solution \cref{eq:snmpc:recursive_feasbility:proof:candiate} as shown in \citet[Proof for \abvThm1, Part III\,–\kern1ptV]{Koehler2019}
	
	\paragraph{Satisfaction of Deterministic ICC Replacement} 
	In the following, we show that the deterministic constraints \cref{eq:snmpc:constraint:chance} used in place of the ICCs \cref{eq:iccs} hold for $k=0,\dotsc,N-1$.
			
	For $k=0,\dotsc,N-2$, we have
	\begin{equation*}
	\begingroup\arraycolsep=1pt\begin{array}{cl}
		&h_j(x_{k+1|t+1},u_{k+1|t+1})+s^p_{k+1|t+1}\\
		\stackrel{\cref{eq:incremental_bound:icc}, \cref{eq:snmpc:recursive_feasbility:proof:Vdelta:bound}}{≤} 
		&h_j(x^*_{k+2|t},u^*_{k+2|t})+c^\mathrm{P}_j ρ^{k+1} w^*_{0|t}+s^p_{k+2|t+1}\\
		\stackrel{\cref{eq:snmpc:recursive_feasbility:proof:sp:decrease}}{≤}{}
		&h_j(x^*_{k+2|t},u^*_{k+2|t})+s^{p,*}_{k+2|t}
		\stackrel{\cref{eq:snmpc:constraint:chance}}{≤} 0
	\end{array}\endgroup 
	\end{equation*}
	The terminal condition \cref{eq:snmpc:constraint:terminal} ensures constraint satisfaction for $k=N-1$ with
	\begin{align*}
	\begingroup\arraycolsep=1pt\begin{array}{cl}
		&h_j(x_{N|t+1},u_{N|t+1})+c^p_js^p_{N|t+1}\\
		\stackrel{
			\cref{eq:snmpc:recursive_feasbility:proof:Vdelta:bound},
			\cref{eq:incremental_bound:icc},
			\cref{eq:snmpc:recursive_feasbility:proof:sp:decrease}
		}{≤}
		&h_j(x^*_{N+1|t},u^*_{N+1|t})+c^p_js^{p,*}_{N+1|t}
		\stackrel{\cref{eq:snmpc:terminalcontroller:icc}}{≤}
		0
	\end{array}\endgroup
	\end{align*}
		
	\paragraph{Practical Stability} 
	As shown in \citet[Proof of Thm.\,1,Part VI]{Koehler2019}, there exist $α^-$, $α^+$, $α_w ∈ \mathcal{K}_∞$ such that
	\begin{gather}
		α^-(‖x(t)‖)≤V_N(x(t))≤α^+(‖x(t)‖)\,,\\
		V_N(x(t+1))-V_N(x(t))≤-α^-(‖x(t)‖)+α_w(\bar{w})\,.
	\end{gather}
	Thus, the closed-loop is practically asymptotically stable. 
	
	\paragraph{Closed-looped Chance Constraint Satisfaction}
	Since \cref{eq:snmpc:constraint:chance} implies the ICCs \cref{eq:iccs} by \cref{thm:backoff:incremental}, the ICCs are satisfied in closed-loop with at least the specified probability.
	\hfill{\qedsymbol}
\endgroup\end{proof}
\begin{proof}[Proof of \cref{thm:backoff:incremental}]
	Again, we consider just a single ICC \cref{eq:icc:single}.
	By \cref{prop:incremental_bound:icc}, $V_δ(x_{k|t+1},x_{k+1|t},u_{k+1|t})≤c^2$ implies 
	\begin{align}
		h(x_{k|t+1},u_{k|t+1})-h(x_{k+1|t},u_{k+1|t}) ≤c^\mathrm{P}\! \cdot  c
		\label{eq:backoff:incremental:derivation:icc:robust:ilyapdiff:predicition}
	\end{align}
	By \cref{ass:local_incremental_stabilizability} \cref{eq:local_incremental_stabilizability:contractivity}, we have  $V_δ(x_{k|t+1},x_{k+1|t},u_{k+1|t}) ≤ (ρs^1_{k|t})^2$ for $d_w=0$.
	Using \cref{prop:disturbance:ilyap}, we can bound the additional increase of $V_δ$ due to the $d_w≠0$. 
	Together, this yields 
	\begin{equation}\bgroup
		\newcommand{\nemu}{\mkern-1mu}\newcommand{\nemumu}{\mkern-2mu}
		ℙ\!\nemu\left[c \nemumu= \smash{\!\sqrt{V_δ(\smash{x_{k|t+1},x_{k+1|t},u_{k+1|t}})}} \nemumu≤\nemumu ρs^1_{k|t}\nemumu+\nemumu w^p_{k|t}\right]\! ≥\nemu p\,.
		\label{eq:backoff:incremental:derivation:s+:probabilistic bound:predicition}
	\egroup\end{equation}
	Then, given that $V_δ(x_{k-1|t+1},v_{k|t},u_{k|t}) ≤ (s^1_{k|t})^2$, substituting \cref{eq:backoff:incremental:derivation:s+:probabilistic bound:predicition} into \cref{eq:backoff:incremental:derivation:icc:robust:ilyapdiff:predicition} yields 
	\begin{equation}\bgroup
			\newcommand{\nemu}{\mkern-1mu}\newcommand{\nemumu}{\mkern-2mu}
			ℙ\!\nemu\left[h(x_{k|t\nemu+\nemu1},u_{k|t\nemu+\nemu1}) \!≤\nemumu h(x_{k\nemu+\nemu1|t},u_{k\nemu+\nemu1|t})\nemumu+\nemumu c^\mathrm{P} s^p_{k\nemu+\nemu1|t}\right]\!≥\nemu p\,,
	\egroup\end{equation}
	hence, we obtain that
	\begin{align}
		&h(x_{k+1|t},u_{k+1|t})+c^\mathrm{P} \cdot  s^p_{k+1|t}≤0\\
		&\qquad⟹ ℙ[ h(x_{k|t+1},u_{k|t+1}) ≤0 ]≥p ⟺ \cref{eq:icc:single}\,.\tag*{\qedsymbol}
	\end{align}
\end{proof}
\subsection{Discussion}
In the following, we discuss various properties of the proposed SMPC scheme, as well as relations to other existing MPC schemes for uncertain systems.
\begin{remark}
Compared to a nominal MPC scheme, $s^p$ and $w^p$ augment the state and the input, resp., for each $p∈\mathcal{P}∪\lbrace 1\rbrace$.
Thus, the online computational demand of solving \cref{eq:snmpc} is comparable to a nominal MPC scheme with $n+1+|\mathcal{P}|$ states, $m+1+|\mathcal{P}|$ inputs, and additional $1+|\mathcal{P}|$ nonlinear constraints for each time step.
Correspondingly, while it is possible to have multiple probability levels $p_j$ for the ICCs, using the same $p$ for all ICCs will be computationally significantly cheaper.
\end{remark}
\begin{remark}
The disturbance bounds \cref{eq:snmpc:var:wp} could also be stated, as equality constraints, which is the case consider for the candidate solution.
Yet, by using inequality constraints, while we still achieve at least the required constraint tightening, we obtain more freedom for the optimization.
Furthermore, allowing tightening beyond the required amount has no impact on the result due to optimality.
In particular, one common form for the function $\tilde{w}_δ$ employs online maximization over a finite set of values, and hence is easily restated as inequality constraints, thereby alleviate the need for additional optimizations within the constraints.
\end{remark}
\begin{remark}
In absence of any ICCs or equivalently, if all ICCs have to be fulfilled with certainty, \ie $\mathcal{P}=\lbrace1\rbrace$, the proposed SMPC scheme trivially reduces to the RMPC scheme proposed by \citet{Koehler2019}.
\end{remark}
\begin{remark}\label[remark]{rem:lipschitz}
The SMPC method in \cite{Santos2019} for additive uncertainty is based on a $\mathcal{K}$-function or as a special case on a Lipschitz constant $L$.
Therein, the authors employed the inverse of empirical cumulative distributions $\hat{F}_W$ as uncertainty description, which can be equivalently used in our method by setting $\hat{w}^p=\hat{F}_W^{-1}(p)$.
In particular, the result on Lipschitz constants $L$ are contained as a special case in our framework with $V_δ(x,z,v)=‖x-z‖^2$, $κ(x,z,v)=v$, and $ρ=L$.
Similarly, the use of a $\mathcal{K}$-function $σ$ is also a special of our scheme using the same $V_δ$ and $κ$ by rewriting $ρ$ and $\tilde{w}_δ$ in terms of $σ$.
The main difference is that we use a backoff for the ICCs depending on the prediction in order to consider state- and input-dependent disturbances, whereas \cite{Santos2019} can employ constant backoff as the current state has no impact on the considered additive stationary stochastic uncertainty.
Overall, the proposed framework using ILF is typically less conservative then either one of these choices \citep{Koehler2018}.
\end{remark}
\begin{remark}
While other methods (\eg \cite{Santos2019}), often require prestabilization in order to limit the tube growth for unstable system.
In our method, this is not necessary as the tube size depends on the controller $κ$ in ILF, which is efficiently able to handle instability.
\end{remark}
\section{Numerical Simulation}\label{sec:numsim}
A widely used benchmark case study in the SMPC literature is the DC-DC-converter regulation problem. 
The discrete-time dynamics translated to the origin are described by \citet{Lazar2008}, including their parameters, as
\begin{equation}
\begingroup\def\arraystretch{1.3}
x^+ \!=\! \begin{bmatrix}
x_1^+\vphantom{\left(\tfrac{T}{L}\right)}\\
x_2^+\vphantom{\left(\tfrac{T}{RC}\right)}
\end{bmatrix} \!=\! \begin{bmatrix}
x_1+αx_2+\left(β-\tfrac{T}{L}x_2\right)u\\
\left(\tfrac{T}{C}x_1+γ\right)u+\left(1-\tfrac{T}{RC}\right)x_2+δx_1
\end{bmatrix}
\endgroup.\label{eq:numsim:sys}
\end{equation}
We consider a (possibly time-varying) parameter uncertainty in $θ=[α,δ]$ with a Gaussian distribution with $Σ_θ=0.1\cdot \mathrm{I}_{2×2}$ variance truncated to a maximal deviation of $1.6σ$.
The system is subject to hard input constraint $|u|≤0.2$, and the electric power is chance constrained by $ℙ[|x_1x_2|≤2]≥0.8$.
Using an ILF $V_δ(x,z,v)=‖x-z‖_P^2$ and controller $κ(x,z,v)=K(x-z)+v$ a contraction rate $ρ≈0.82$ can be achieved.
The disturbance bounds $\tilde{w}_δ^p(x,u,c) = \left\lVert P^\frac{1}{2}\frac{∂f(x,u)}{∂[α,δ]}Σ_θ^\frac{1}{2}\right\rVert ε(p)+L_wc$ is derived analogously to \cite[\abvProp3]{Koehler2019} with Lipschitz constant $L^1_w≈0.15$, $L^{0.6}_w≈0.06$ and $ℙ[‖θ‖^2_{Σ_θ^{-1}}≤ε(p)]≥p$.

For a standard quadratic cost in state and input, stability and constraints satisfaction is achieved.
For this example, however, we shall consider $ℓ(x,u)=u^2$, which voids the stability claim as \cref{ass:cost:kinf_bound} is violated.
We exploit that inherent instability of the system \cref{eq:numsim:sys} will drive the system into the constraint.
Thereby, we can more easily study the chance constraint satisfaction.
In the same vein, we drop the terminal constraints required for the stabilization.

\begin{figure}
    \bgroup
    \def\mydatapath{./figs}
    	\begin{tikzpicture}[%
        /pgfplots/every major grid/.style={black!40},%
        /pgfplots/every minor grid/.style={black!10},%
    ]
		\begin{axis}[
            grid=both,minor tick num=1,distance=1,
            xmin=-2,xmax=2,ymin=-3,ymax=4,samples=200,
            xlabel={$x_1$},ylabel={$x_2$}
        ]
            \addplot[smooth,red,thick] {min(sqrt(2/x/x),5)};
            \addplot[smooth,red,thick] {max(-sqrt(2/x/x),-3.5)};
            \addplot[] table[x=x1,y=x2] {\mydatapath/numsim.csv} node[pos=0] (x0) {};%
		\end{axis}
	\end{tikzpicture}
    \egroup
	\caption{Evolution of closed-loop states over 1000 time-steps. The red lines represented the chance constraint. The initial condition is $x(0)=[-1.25,0]^⊤$.}
	\label{fig:numsim}
\end{figure}
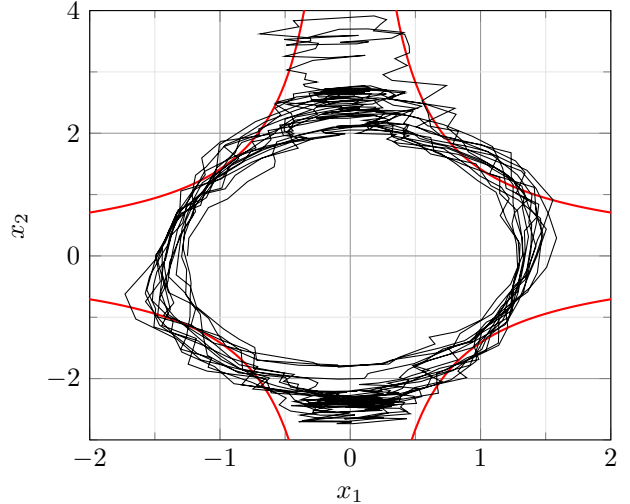

In \cref{fig:numsim}, we can see one evolution of the closed-loop under the proposed SMPC without stabilization.
In 95\% of the 14000 simulated steps the power constraint is satisfied. 
While the constraint is not active in all of the considered state (\cf \cref{fig:numsim}), this is admittedly still rather conservative compared to the required 80\% satisfaction.
Yet, this is significantly less conservative then approximating the disturbance as additive would be.
Likewise, as stated in \cref{rem:lipschitz}, using Lipschitz constants would yield a more conservative result.
With our approach, however, the conservatism could be further reduced by using a less conservative disturbance bound $\tilde{w}_δ$ at the price of additional computational complexity.
Therefore, our method can be tune to the desired comprise between conservatism and complexity.
At the same time with the fairly conservative, but simple, bound presented in this example, we achieved tighter satisfaction than existing methods.

\section{Conclusion}
We have proposed a nonlinear SMPC framework based on incremental stabilizability for nonlinear systems incorporating general state- and  input- dependent uncertainty descriptions.
The scheme can ensures satisfaction of individual chance constraints and hard constraints, as well as recursive feasibility.
By using a specially designed growing tube, we achieve this with only a small computation cost increase over nominal MPC.
We have demonstrated the performance gains of the proposed framework over RMPC. 

\begin{acknowledgments}
The authors would like to thank our colleagues J.\,Köhler, and R.\,Soloperto for many helpful comments and fruitful discussions.
\end{acknowledgments}

\bibliography{database}

\begin{thebibliography}{13}
\providecommand{\natexlab}[1]{#1}
\providecommand{\url}[1]{\texttt{#1}}
\providecommand{\urlprefix}{URL }
\expandafter\ifx\csname urlstyle\endcsname\relax
  \providecommand{\doi}[1]{doi:\discretionary{}{}{}#1}\else
  \providecommand{\doi}{doi:\discretionary{}{}{}\begingroup
  \urlstyle{rm}\Url}\fi

\bibitem[{Chisci et~al.(2001)Chisci, Rossiter, and Zappa}]{Chisci2001}
Chisci, L., Rossiter, J., and Zappa, G. (2001).
\newblock Systems with persistent disturbances: predictive control with
  restricted constraints.
\newblock \emph{Automatica}, 37(7), 1019--1028.
\newblock \doi{10.1016/S0005-1098(01)00051-6}.

\bibitem[{{K{\"{o}}hler} et~al.(2018){K{\"{o}}hler}, {M{\"{u}}ller}, and
  {Allg{\"{o}}wer}}]{Koehler2018}
{K{\"{o}}hler}, J., {M{\"{u}}ller}, M.A., and {Allg{\"{o}}wer}, F. (2018).
\newblock A novel constraint tightening approach for nonlinear robust model
  predictive control.
\newblock In \emph{Proceedings of the American Control Conference}, 728--734.
\newblock \doi{10.23919/ACC.2018.8431892}.

\bibitem[{K{\"{o}}hler et~al.(2019)K{\"{o}}hler, Soloperto, M{\"{u}}ller, and
  Allg{\"{o}}wer}]{Koehler2019}
K{\"{o}}hler, J., Soloperto, R., M{\"{u}}ller, M.A., and Allg{\"{o}}wer, F.
  (2019).
\newblock A computationally efficient robust model predictive control framework
  for uncertain nonlinear systems.
\newblock Preprint \arxiv{1910.12081v1}.

\bibitem[{Kouvaritakis and Cannon(2016)}]{Kouvaritakis2016}
Kouvaritakis, B. and Cannon, M. (2016).
\newblock \emph{Model Predictive Control: Classical, Robust and Stochastic}.
\newblock Springer.
\newblock \doi{10.1007/978-3-319-24853-0}.

\bibitem[{Lazar et~al.(2008)Lazar, Heemels, Roset, Nijmeijer, and van~den
  Bosch}]{Lazar2008}
Lazar, M., Heemels, W.P.M.H., Roset, B.J.P., Nijmeijer, H., and van~den Bosch,
  P.P.J. (2008).
\newblock Input-to-state stabilizing sub-optimal {NMPC} with an application to
  {DC}--{DC} converters.
\newblock \emph{International Journal of Robust and Nonlinear Control}, 18(8),
  890--904.
\newblock \doi{10.1002/rnc.1249}.

\bibitem[{Mayne(2014)}]{Mayne2014}
Mayne, D.Q. (2014).
\newblock Model predictive control: Recent developments and future promise.
\newblock \emph{Automatica}, 50(12), 2967--2986.
\newblock \doi{10.1016/j.automatica.2014.10.128}.

\bibitem[{Mesbah(2016)}]{Mesbah2016}
Mesbah, A. (2016).
\newblock Stochastic model predictive control: An overview and perspectives for
  future research.
\newblock \emph{{IEEE} Control Systems Magazine}, 36(6), 30--44.
\newblock \doi{10.1109/MCS.2016.2602087}.

\bibitem[{{Pin} et~al.(2009){Pin}, {Raimondo}, {Magni}, and
  {Parisini}}]{Pin2009}
{Pin}, G., {Raimondo}, D.M., {Magni}, L., and {Parisini}, T. (2009).
\newblock Robust model predictive control of nonlinear systems with bounded and
  state-dependent uncertainties.
\newblock \emph{{IEEE} Transactions on Automatic Control}, 54(7), 1681--1687.
\newblock \doi{10.1109/TAC.2009.2020641}.

\bibitem[{Rawlings et~al.(2017)Rawlings, Mayne, and Diehl}]{Rawlings2017}
Rawlings, J.B., Mayne, D.Q., and Diehl, M. (2017).
\newblock \emph{Model predictive control: theory, computation, and design}.
\newblock Nob Hill Publishing.

\bibitem[{{Santos} et~al.(2019){Santos}, {Bonzanini}, {Heirung}, and
  {Mesbah}}]{Santos2019}
{Santos}, T.L.M., {Bonzanini}, A.D., {Heirung}, T.A.N., and {Mesbah}, A.
  (2019).
\newblock A constraint-tightening approach to nonlinear model predictive
  control with chance constraints for stochastic systems.
\newblock In \emph{Proceedings of the American Control Conference}, 1641--1647.

\bibitem[{Schildbach et~al.(2014)Schildbach, Fagiano, Frei, and
  Morari}]{Schildbach2014}
Schildbach, G., Fagiano, L., Frei, C., and Morari, M. (2014).
\newblock The scenario approach for stochastic model predictive control with
  bounds on closed-loop constraint violations.
\newblock \emph{Automatica}, 50(12), 3009--3018.
\newblock \doi{10.1016/j.automatica.2014.10.035}.

\bibitem[{{van Hessem} and {Bosgra}(2002)}]{vanHessem2002}
{van Hessem}, D.H. and {Bosgra}, O.H. (2002).
\newblock A conic reformulation of model predictive control including bounded
  and stochastic disturbances under state and input constraints.
\newblock In \emph{Proceedings of the {IEEE} Conference on Decision and
  Control.}, vol.\,4, 4643--4648.
\newblock \doi{10.1109/CDC.2002.1185110}.

\bibitem[{Villanueva et~al.(2017)Villanueva, Quirynen, Diehl, Chachuat, and
  Houska}]{Villanueva2017}
Villanueva, M.E., Quirynen, R., Diehl, M., Chachuat, B., and Houska, B. (2017).
\newblock Robust {MPC} via min--max differential inequalities.
\newblock \emph{Automatica}, 77, 311--321.
\newblock \doi{10.1016/j.automatica.2016.11.022}.

\end{thebibliography}
\normalsize
\end{document}